\documentclass[sigconf]{acmart}
\AtBeginDocument{%
  \providecommand\BibTeX{{%
    \normalfont B\kern-0.5em{\scshape i\kern-0.25em b}\kern-0.8em\TeX}}}

\setcopyright{acmcopyright}
\copyrightyear{2023}
\acmYear{2023}
\acmDOI{XXXXXXX.XXXXXXX}

\acmConference[KDD 2023 Workshop Causal Inference and Machine Learning in Practice]{}{August 7,
  2023}{Long Beach, CA}
\usepackage{natbib}
\usepackage{cleveref}
\usepackage{hyperref}
\usepackage{mathtools}

\newcommand{\ra}{\rightarrow}
\newcommand{\la}{\leftarrow}
\newcommand{\bern}{\mathrm{Bern}}
\newcommand{\rank}{\mathrm{rank}}

\newtheorem{assumption}{Assumption}
\allowdisplaybreaks
\begin{document}

\title{An IPW-based Unbiased Ranking Metric in Two-sided Markets}

\author{Keisho Oh}
\affiliation{%
  \institution{Recruit Co., Ltd.}
  \city{Tokyo}
  \country{Japan}}
\email{oh_keisho@r.recruit.co.jp}

\author{Naoki Nishimura}
\affiliation{%
  \institution{Recruit Co., Ltd.}
  \city{Tokyo}
  \country{Japan}}
\email{nishimura@r.recruit.co.jp}

\author{Minje Sung}
\affiliation{%
  \institution{Tokyo Institute of Technology}
  \city{Tokyo}
  \country{Japan}}
\email{sung.m.aa@m.titech.ac.jp}

\author{Ken Koboyashi}
\affiliation{%
  \institution{Tokyo Institute of Technology}
  \city{Tokyo}
  \country{Japan}}
\email{kobayashi.k.ar@m.titech.ac.jp}

\author{Kazuhide Nakata}
\affiliation{%
  \institution{Tokyo Institute of Technology}
  \city{Tokyo}
  \country{Japan}}
\email{nakata.k.ac@m.titech.ac.jp}


\begin{abstract}
In modern recommendation systems, unbiased learning-to-rank (LTR) is crucial for prioritizing items from biased implicit user feedback, such as click data. 
Several techniques, such as Inverse Propensity Weighting (IPW), have been proposed for single-sided markets.
However, less attention has been paid to two-sided markets, such as job platforms or dating services, where successful conversions require matching preferences from both users.
This paper addresses the complex interaction of biases between users in two-sided markets and proposes a tailored LTR approach. 
We first present a formulation of feedback mechanisms in two-sided matching platforms and point out that their implicit feedback may include position bias from both user groups. 
On the basis of this observation, we extend the IPW estimator and propose a new estimator, named \emph{two-sided IPW}, to address the position bases in two-sided markets. 
We prove that the proposed estimator satisfies the unbiasedness for the ground-truth ranking metric. 
We conducted numerical experiments on real-world two-sided platforms and demonstrated the effectiveness of our proposed method in terms of both precision and robustness. 
Our experiments showed that our method outperformed baselines especially when handling rare items, which are less frequently observed in the training data.

\end{abstract}





% \received{20 December 2022}

\maketitle
\section{Introduction}
Information retrieval and recommendation systems have become increasingly prevalent. 
One of the key features of these systems is that they present users with items ranked in accordance with their preferences.
To this end, such systems often use learning-to-rank (LTR) algorithms that capture users' preferences, prioritizing items that are more likely to be relevant or appealing to the user. 
As the number of available items continues to grow, LTR has become an essential component of effective recommendation systems, attracting considerable interest from both industry and academic fields.

One of the key challenges of LTR is \emph{unbiased learning}, which aims to accurately estimate users' preferences without being influenced by biases in training data. 
Usually, LTR methods rely on users' implicit feedback, such as clicks, because this type of data can be easily collected during users' natural behaviors in the system.
However, such implicit feedback sometimes reflects not only users' true preferences but also other factors like their first impressions of items or the item's position in the list~\cite{Wen2019}. 
As a result, learning users' preferences from implicit feedback can be biased and vulnerable, which leads to suboptimal recommendations and a degraded user experience.
To address this issue, the unbiased LTR framework has attracted much attention~\cite {Chen2023}.  
This research topic was first introduced by \citet{Joachims2017}, and several unbiased LTR methods have been proposed~\cite{Wang2016,Ai2018,wang2018pbm,hu2019unbiased_lambdamart} to improve the quality of recommendation systems. 

While several unbiased learning techniques have been proposed in the literature on LTR, existing methods mainly focus on single-sided markets, and less attention has been paid to two-sided markets.
In two-sided markets, the final conversion is determined only when preferences from both users match.   
In such a situation, position biases are included in the implicit feedback not only from one-side users but also from the opposite users.  
Let us consider the case of a two-sided job-matching platform for example. 
A job seeker first inputs a query to search for job postings matching his/her preference.  
Then, the job seeker browses the list of job postings as a result and clicks on several of them to confirm the more detailed information or submit an application~\cite{Chen2019}. 
On the other hand, a recruiter browses a list of applicants to choose some promising ones. 
Finally, the recruiter makes a final decision on which applicant to hire. 
In such a scenario, both the job seeker and recruiter tend to recognize items presented at higher positions. 
As a result, there exists not only a position bias against job seekers but also one that recruiters are more likely to hire the candidates who are displayed at the top of the list.  

Against this background, we propose an unbiased LTR framework for two-sided matching platforms. 
In this work, we first present an estimator called the two-sided inverse probability weighting (IPW) estimator by extending the (IPW) estimator~\citep{Joachims2017}. 
We prove that the proposed metric is unbiased where biased is included in implicit feedback from the two sides in matching platforms. 
Also, we  present a practical LTR algorithm with the two-sided IPW estimator. 
To the best of our knowledge, we are the first to propose an unbiased LTR method in two-sided matching platforms.  

The contributions of our study can be summarized as follows:
\begin{itemize}
    \item We are the first to give a formulation of feedback mechanisms in two-sided matching platforms and point out that their implicit feedback may include position bias from both user groups.
    \item To address these biases, we propose a new estimator, named \emph{two-sided IPW} by extending the IPW estimator. 
    \item We also proved that the proposed estimator is unbiased against position biases in two-sided markets. 
    \item  We conducted numerical experiments on both synthetic and real datasets and demonstrated that the LTR method with our estimator significantly outperformed existing methods. 
\end{itemize}

\section{Related Work}
\subsection{Learning-to-Rank}
LTR refers to a framework for constructing ranking models where items are sorted in accordance with their degrees of relevance or preferences of users. 
LTR was first introduced by \citet{Fuhr1989} and due to its practical importance, this topic has attracted much attention.  
The approaches of existing LTR methods are categorized into three types: pointwise, pairwise, and listwise approaches. 
The pointwise approaches~\citep{Crammer2001,Cooper1992,Li2007,Nallapati2004,Chu2007,Cossock2006} focus on determining whether a single item is relevant or not by itself. 
They train a classification or regression model to predict whether the individual item is relevant to the input query.
On the other hand, the pairwise approaches~\citep{burges2010ranknet,Freund2003,Joachims2006,Xu2007,Shen2005} use pairs of items as a training instance and train a ranking function from the relative order of relevance between the items in each pair. 
The listwise methods~\citep{cao2007listnet,Guiver2008,Qin2009,Taylor2008,Volkovs2009,Wu2009,Xia2008,Xu2007,Yue2007,burges2010ranknet,bruch2021xendcg} aim to train a ranking function where the ranked lists of items are given. 
Contrary to the pointwise and pairwise approaches, the listwise approaches can capture the nature of the ranking problems, and the previous studies showed that the listwise approaches delivered better performance than the others~\citep{cao2007listnet,Qin2008,Qin2009}.

\subsection{Unbiased Learning-to-rank Algorithms}
The Unbiased LTR framework has attracted much attention as an effective approach to train a ranking model from biased feedback for position bias~\cite {Chen2023}.  
\citet{Joachims2017} first dealt with the unbiased ranking problems and presented the counterfactual framework to learn an unbiased ranking function by using the IPW estimator. 
For the propensity weights, intervention-based and result randomization-based estimation methods have been proposed~\cite{Joachims2017,Wang2016}. 
\citet{Ai2018} devised a method called DLA that simultaneously trains a ranking function and estimates propensity weights. 
\citet{wang2018pbm} proposed a method for position bias estimation with a regression-based expectation-maximization algorithm. 
Furthermore, \citet{hu2019unbiased_lambdamart} proposed a listwise unbiased LTR method called LambdaMart.

As previously mentioned, several methods have been proposed for training an unbiased ranking function.
However, these existing studies focus on cases where the bias is derived from the implicit feedback from only one-side users. 
In contrast, our proposed method addresses the cases where the implicit feedback includes not only the bias derived from one side but also from the opposite side.

\subsection{Recommendation in Matching Platforms}
As matching platforms have been emerging in various fields such as job recruiting and  dating,  two-sided recommendation methods have begun to be actively studied. 
\citet{yang2022jobmatch} used Bayesian personalized ranking (BPR) loss function to predict the two-sided preferences taking into account the expectation from both sides. 
\citet{Tu2014} proposed a Latent dirichlet allocation (LDA) probabilistic model to learn the preferences in two-sided dating platforms. 
\citet{su2022opt_rank} focused on the ranking problem in two-sided matching and proposed an optimization problem to maximize the overall social welfare. 
In addition, \citet{saito2022fair} defined a ranking optimization problem to maximize the Nash Social Welfare and provided its tractable formulation. 
However, these existing studies do not address the issue that the implicit feedback includes bias from both sides and thus, these methods are expected not to perform well under a biased dataset.

\section{Proposed Method}
We first present a formulation of feedback mechanisms in matching platforms and point out it potentially has two types of position bias.
Then, we address how to debias these two types of position biases in ranking metrics. 
Finally, we derive an unbiased LTR method for two-sided matching platforms.

\subsection{Notation and Problem Setting}
We consider a platform that has two sets of users, a proactive set $U$ and a reactive set $V$.
We assume two types of users that interact with each other as follows: the platform first presents each proactive user $u \in U$ with a personalized ranking list $V_u \subseteq V$ that includes reactive users who the user $u$ might be interested in.
In accordance with the ranking list, each proactive user $u$ selects specific reactive users on the basis of his/her preferences.
Next, for each reactive side user $v \in V$, the platform shows a new ranking list $U_v\subseteq U$ including proactive users who have already selected $v$. 
We define $\rank_{u}(v)$ as the ranking of user $v$ in the list shown to $u$ and $\rank_{v}(u)$ as the ranking of user $u$ conversely. 

To express implicit interactions from the proactive user $u\in U$ to the reactive user $v\in V$, we use a binary random variable $Y_{u \ra v}$. 
If $Y_{u \ra v}=1$, the proactive user $u$ gives positive feedback to the reactive user $v$; otherwise $Y_{u \ra v}=0$. 
We also introduce a binary random variable $Y_{u \la v}$ to express the implicit feedback from the reactive user $v$ to the proactive user $u$. 
To denote the relevance from the proactive user $u$ to the reactive user $v$, we use a binary label $R_{u \ra v}\in \{0,1\}$. 
Conversely, $R_{u \la v}$ represents the relevance from $v$ to $u$. 
If $R_{u \ra v}=1$ (resp. $R_{u \la v}=1$), the reactive user $v$ (resp. the proactive user $u$) is relevant to the proactive user $u$ (resp. the reactive user $v$). 
We denote the exposure variable as $O_{u \ra v}$ (resp. $O_{u \la v}$) whether the reactive user $v$ (resp. the proactive user $u$) is exposed to the proactive user $u$ (resp. the reactive user $v$).  
Let us denote $\theta_{u \ra v}, \theta_{u \la v}$ the probability of $O_{u \ra v}, O_{u \la v}$ get $1$ respectively. 
We summarize the notation in \Cref{table:notation}. 

\begin{table}[ht]
  \caption{Summary of notation}
  \label{table:notation}
  \centering
  \begin{tabular}{cl}
    \hline
    Notation  & Explanation \\
    \hline \hline
    $u, u_{i}$ & proactive user \\
    $v, v_{j}$ & reactive user \\
    $U$ & set of proactive users \\
    $V$ & set of reactive users \\
    $U_v$ & set of proactive users shown to $v$ \\
    $V_u$ & set of reactive users shown to $u$ \\
    $\rank_{u}(v)$ & the rank of $v$ at the list shown to $u$ \\
    $Y_{u \ra v}$ / $Y_{u \la v}$ & proactive/reactive user's feedback \\
    $R_{u \ra v}$ / $R_{u \la v}$ & true proactive/reactive side-specific relevance\\
    $O_{u \ra v}$ / $O_{u \la v}$ & proactive/reactive exposure indicator\\
    $\theta_{u \ra v}$ / $\theta_{u \la v}$ & proactive/reactive exposure probability\\
    \hline
  \end{tabular}
\end{table}

In this study, we rely on the following two technical assumptions, which are often used in causal inference:
\begin{assumption}
    For all proactive user $u\in V$, the exposure probability $\theta_{u\ra v}$ only depends on the position of $v\in V$ in the list, i.e., $O_{u \ra v} \sim P(\cdot \mid \rank_{u}(v))$. 
\end{assumption}
\begin{assumption}
For any proactive-reactive user pair $(u,v)\in U\times V$, $\theta_{u\ra v},  \theta_{u\la v}\in (0,1]$.
\end{assumption}

We can now model implicit feedback mechanisms in the two-sided matching platform. 
For proactive users, $u \in U$ selects $v \in V$ only if $u$ examined $v$ ($O_{u\ra v}=1$) and found it relevant ($R_{u\ra v}=1$). 
For reactive users, $v\in V$ selects $u\in U$ only if the following conditions are all satisfied:
\begin{enumerate}
    \item The proactive user $u$ has already selected $v$ ~($Y_{u\ra v}=1$);
    \item The proactive user $u$ has been exposed to $u$~($O_{u\la v}=1$);
    \item The proactive user $u$ is relevant to $v$~~($R_{u\ra v}=1$). 
\end{enumerate} 
In this setting, the implicit feedback model can be expressed as follows: 
\begin{align}
  Y_{u \ra v} &= O_{u \ra v} \cdot R_{u \ra v}, \label{feedback:reactive} \\
  Y_{u \la v} &= Y_{u \ra v} \cdot O_{u \la v} \cdot R_{u \la v}\label{feedback:proactive}.
\end{align}

\subsection{Ranking Metric and Naive Estimator}
To evaluate the ranking $\rank_{u}(v)$ in the two-sided matching platform, we define a \emph{two-sided relevance} as $R_{uv} \coloneq R_{u \ra v}(1 + R_{u \la v})$. 
The two-sided relevance $R_{uv}$ represents the following three scenarios:
\begin{align*}
    R_{uv} = \begin{cases}
        2&\text{if~$R_{u \ra v}=1$ and~$R_{u \la v}=1$},\\
        1&\text{if~$R_{u \ra v}=1$ and~$R_{u \la v}=0$},\\\
        0&\text{otherwise}.
    \end{cases}
\end{align*}
Then, we can define the ground-truth performance of a given  ranking $ \rank_{u}(v)$  as follows:
\begin{align}
\label{metric:ground_truth}
&\mathcal{R}\left(\rank_{u}(v), R_{u\ra v}, R_{u \la v}\right) \nonumber \\
&\quad \coloneq \frac{1}{\left| U \right|}\sum_{u \in U} \sum_{v \in V_u} \lambda(\rank_{u}(v)) \cdot g(R_{u \ra v}, R_{u \la v}). 
\end{align}
Here, $g(R_{u \ra v}, R_{u \la v})$ is the gain function defined as
$$
g(R_{u \ra v}, R_{u \la v})\coloneq  2^{R_{uv}} - 1 =2^{R_{u \ra v}(1 + R_{u \la v})}-1,
$$
and $\lambda\left(\rank_{u}(v)\right)$ is a weight function determining the ranking metric. 
For instance, when we set 
$$\lambda\left(\rank_{u}(v)\right) = \frac{\mathbf{1}(\rank_{u}(v) \le K)}{\log_2\left(\rank_{u}(v) + 1\right)},$$ 
where $\mathbf 1(\cdot)$ is the indicator function and $K$ is a user-defined positive integer, \Cref{metric:ground_truth} corresponds to  $\mathrm{DCG}@K$.

Since we cannot directly observe the relevance $R_{u\ra v}$ and $R_{u\la v}$, we consider designing an estimator for the ground-truth metric with the observable implicit feedback $Y_{u\ra v}$ and $Y_{u\la v}$.  
Let us define a surrogate gain function with $Y_{u\ra v}$ and $Y_{u\la v}$ as
\begin{align}\label{eq:surrogate_gain}
h(Y_{u \ra v}, Y_{u \la v}) &:= 2^{Y_{u \ra v} + Y_{u \la v}} - 1. 
\end{align}
We note that from  \Cref{feedback:reactive,feedback:proactive}, the following holds
\begin{align*}
    h(Y_{u \ra v}, Y_{u \la v}) &= 2^{O_{u\ra v} R_{u\ra v}(1+ O_{u\la v} R_{u\la v})} \\
    &= g(O_{u\ra v} R_{u\ra v}, O_{u\la v} R_{u\la v}). 
\end{align*}
Thus, we can regard $h(Y_{u \ra v}, Y_{u \la v})$ as an alternative to the gain function $g$ based on the observable implicit feedback $Y_{u\ra v}$ and $Y_{u\la v}$.

Using this surrogate gain function $h$, we can define a naive estimator of \Cref{metric:ground_truth} as follows:
\begin{align}
\label{metric:naive_estimator}
&\hat{\mathcal{R}}\left(\rank_{u}(v), Y_{u\ra v}, Y_{u \la v}\right)\notag \\
&\quad \coloneq \frac{1}{\left| U \right|}\sum_{u \in U} \sum_{v \in V_u} \lambda(\rank_{u}(v)) \cdot h(Y_{u \ra v}, Y_{u \la v}).
\end{align}
However, the naive estimator $\hat{\mathcal{R}}$ is biased as shown below:
\begin{proposition} 
Suppose that $O_{u \ra v}$ and $O_{u \la v}$ are independent of each other and $\theta_{u\ra v}, \theta_{u\la v}\neq 1 $. 
Then, the following holds
\begin{align*}
&\mathbb{E}_{O_{u \ra v}, O_{u \la v}} \left[ \hat{\mathcal{R}}\left(\rank_{u}(v), Y_{u \ra v}, Y_{u \la v}\right) \right]\\ 
&\quad \not\propto  \mathcal{R}\left(\rank_{u}(v), R_{u\ra v}, R_{u \la v}\right).
\end{align*}
\end{proposition}

\begin{proof}
It is sufficient to show that
\begin{align*}
\mathbb{E}_{O_{u \ra v}, O_{u \la v}}\left[h(Y_{u \ra v}, Y_{u \la v})\right] \not\propto g(R_{u \ra v}, R_{u \la v}).
\end{align*}
Since $O_{u \ra v}, O_{u \la v}$ are binary random variables, it holds that
\begin{align}
2^{Y_{u \ra v}} &= 2^{O_{u \ra v} R_{u \ra v}} \nonumber \\
&= (1 - O_{u \ra v}) + O_{u \ra v} 2^{R_{u \ra v}}, \label{eq:proactive-exp-gain} \\
2^{Y_{u \la v}} &= 2^{O_{u \ra v} O_{u \la v} R_{u \ra v} R_{u \la v}} \nonumber \\
&= (1 - O_{u \ra v}O_{u \la v}) + O_{u \ra v}O_{u \la v} 2^{R_{u \ra v}R_{u \la v}}.
\end{align}
Then, we have
\begin{align}
2^{Y_{u \ra v}} \cdot 2^{Y_{u \la v}} &= (1 - O_{u \ra v}) + O_{u \ra v}(1 - O_{u \la v})2^{R_{u \ra v}}\\
& \quad + O_{u \ra v}O_{u \la v}2^{R_{u \ra v}(1 + R_{u \la v})}. \label{eq:reactive-exp-gain}
\end{align}
Since $O_{u \ra v}, O_{u \la v}$ are independent and $\theta_{u\ra v}, \theta_{u\la v} \neq 1$, we obtain
\begin{align}
%\label{proof:naive_estimator}
&\mathbb{E}_{O_{u \ra v}, O_{u \la v}}\left[ h(Y_{u \ra v}, Y_{u \la v})\right] \notag \\
&\quad = \mathbb{E}_{O_{u \ra v}, O_{u \la v}}\left[ 2^{Y_{u \ra v}} \cdot 2^{Y_{u \la v}} \right] \notag \\
&\quad = (1 - \theta_{u \ra v}) + \theta_{u \ra v}\left(1 - \theta_{u \la v}\right)2^{R_{u \ra v}}\notag \\ 
&\quad \quad + \theta_{u \ra v} \theta_{u \la v} \underbrace{2^{R_{u \ra v}(1 + R_{u \la v})}}_{g(R_{u \ra v}, ~R_{u \la v}) + 1}, \label{eq:reactive-exp-gain} \\
&\quad \not\propto   g(R_{u \ra v}, R_{u \la v}),\notag
\end{align}
which completes the proof.
\end{proof}

\subsection{Unbiased Ranking Metric Estimator}
We have seen that the naive estimator~\eqref{metric:naive_estimator} can be biased in estimating the ground-truth ranking metric with implicit feedback.  
To remedy this issue, we extend the IPW estimator \citep{Joachims2006} and propose a new estimator, called \emph{two-sided IPW estimator}, for the ranking metric~\eqref{metric:ground_truth}. 

To take into account the position bias in the two-sided matching platform, we first introduce the following weighted gain function: 
\begin{align}
h_{\mathrm{IPW}}(Y_{u \ra v}, Y_{u \la v}) &= \frac{1}{\theta_{u \ra v} \theta_{u \la v}}\left(2^{Y_{u \ra v}}\left(2^{Y_{u \la v}} - 1\right)\right) \nonumber \\
&\quad + \frac{1}{\theta_{u \ra v}} \left(2^{Y_{u \ra v}} - 1\right). \nonumber
\end{align}
With $h_{\mathrm{IPW}}(Y_{u \ra v}, Y_{u \la v})$, we define the \emph{two-sided IPW estimator} as follows:
\begin{equation}\label{metric:ipw_estimator}
\begin{split}
&\hat{\mathcal{R}}_{\mathrm{IPW}}(Y_{u \ra v}, Y_{u \la v})  \\  
&\quad \coloneq \frac{1}{\left| U \right|}\sum_{u \in U} \sum_{v \in V_u} \lambda(\rank_{u}(v)) \cdot h_{\mathrm{IPW}}(Y_{u \ra v}, Y_{u \la v}).    
\end{split}
\end{equation}
We show that $\hat{\mathcal{R}}_{\mathrm{IPW}}$ is an unbiased estimator of the ground-truth metric~\eqref{metric:ground_truth}. 
\begin{theorem}
Suppose that $O_{u \ra v}$ and $O_{u \la v}$ are independent of each other. 
Then it follows that
\begin{equation}\label{eq:unbias}
\begin{split}
     &\mathbb E_{O_{u\ra v}, O_{u\la v}}\left[\hat{\mathcal{R}}_{\mathrm{IPW}}(Y_{u \ra v}, Y_{u \la v})\right] \notag\\
     &\quad = \mathcal{R}\left(\rank_{u}(v), R_{u\ra v}, R_{u \la v}\right). 
\end{split}    
\end{equation}
\end{theorem}
\begin{proof}
From \Cref{eq:proactive-exp-gain}, the expectation of $2^{Y_{u \ra v}}$ can be calculated as
\begin{align*}
\mathbb{E}\left[ 2^{Y_{u \ra v}} \right] &= \mathbb{E}\left[ (1 - O_{u \ra v}) + O_{u \ra v} 2^{R_{u \ra v}} \right] \\
&= (1 - \theta_{u \ra v}) + \theta_{u \ra v} 2^{R_{u \ra v}}.
\end{align*}
Combined with \Cref{eq:reactive-exp-gain},  the expectation of $h_{\mathrm{IPW}}\left(Y_{u \ra v}, Y_{u \la v}\right)$ is calculated as
\begin{align}
&\mathbb{E}_{O_{u \ra v}, O_{u \la v}}\left[ h_{\mathrm{IPW}}\left(Y_{u \ra v}, Y_{u \la v}\right) \right] \notag\\
&~= \mathbb{E}_{O_{u \ra v}, O_{u \la v}}\left[ \frac{1}{\theta_{u \ra v}\theta_{u \la v}} 2^{Y_{u \ra v}} 2^{Y_{u \la v}}\right.\\
&~\quad + \left.\frac{1}{\theta_{u \ra v}}\left(1 - \frac{1}{\theta_{u \la v}}\right) 2^{Y_{u \ra v}} - \frac{1}{\theta_{u \ra v}} \right] \notag\\
&~= \frac{1}{\theta_{u \ra v}\theta_{u \la v}}  \mathbb{E}_{O_{u \ra v}, O_{u \la v}}\left[  2^{Y_{u \ra v}}2^{Y_{u \la v}} \right] \notag\\
&~\quad + \frac{1}{\theta_{u \ra v}}\left(1 - \frac{1}{\theta_{u \la v}}\right) \mathbb{E}_{O_{u \ra v}, O_{u \la v}}\left[  2^{Y_{u \ra v}} \right] - \frac{1}{\theta_{u \ra v}} \notag\\
&~= \frac{1}{\theta_{u \ra v}\theta_{u \la v}} \biggl( (1 - \theta_{u \ra v}) + \theta_{u \ra v}\left(1 - \theta_{u \la v}\right)2^{R_{u \ra v}}  \notag\\ 
&~\quad  + \theta_{u \ra v} \theta_{u \la v} \left(g\left(R_{u \ra v}, ~R_{u \la v}\right) + 1\right) \biggr)\notag \\
&~\quad +  \frac{1}{\theta_{u \ra v}}\left(1 - \frac{1}{\theta_{u \la v}}\right) \left( (1 - \theta_{u \ra v}) + \theta_{u \ra v} 2^{R_{u \ra v}} \right) - \frac{1}{\theta_{u \ra v}}  \notag\\
&~= g\left(R_{u \ra v}, ~R_{u \la v}\right) \label{eq:equaltog}.
\end{align}
From \Cref{metric:ground_truth,metric:ipw_estimator,eq:equaltog}, we obtain the desired result. 
\end{proof}

\section{Experiment}
\subsection{Experimental Setup}
\subsubsection{Data Preparation}
We conducted numerical experiments on the networking recommendation data set collected by \citet{su2022opt_rank}.
The data set provides a matrix expressing the bidirectional preferences within the range $[0, 1]$ of 925 users. 
We denote this matrix as $M \in [0, 1]^{925 \times 925}$ , where $M_{uv}$ represents the preference from $u$ to $v$. 
Then we preprocessed the data as follows.
\begin{enumerate}
    \item To simulate a two-sided market, we randomly divided the user group into a proactive side $U$ and a reactive side $V$. 
    \item For the purpose of cross-validation, we split each $U$ and $V$ into 5-fold, which is denoted by $U_1, \cdots, U_5$ and $V_1 \cdots, V_5$.
    \item We picked the $k$-th fold elements of $\left\{M_{ij}, i \in U_k, j \in V_k \right\}$ as the test data, and split the rest elements as train and validation data set.
    \item We defined the exposure probabilities $\theta_{u \ra v}, \theta_{u \la v}$ as follows:
        \begin{align*}
            &\theta_{\cdot \ra v} =\left(\frac{ \sum_{k \in U} M_{kv}}{ \underset{v' \in V}{\max} \sum_{k \in U} M_{kv'}}\right)^\eta, \\
            &\theta_{u \la \cdot} =\left(\frac{ \sum_{k \in V} M_{ku}}{ \underset{u' \in U}{\max} \sum_{k \in V} M_{ku'}}\right)^\eta,
        \end{align*}
        where $\eta$ is a hyperparameter. 
        Here, the exposure probability for each user is determined by the proportion of interactions of the user compared to the one with the most interactions.
        % 直感的な説明を追加
    \item We randomly generated feedback variables as follows:
        \begin{align*}
            &R_{u \ra v} \sim \bern(M_{u \ra v}), O_{u \ra v} \sim \bern(\theta_{u \ra v}), \\
            &R_{u \la v} \sim \bern(M_{u \la v}), O_{u \la v} \sim \bern(\theta_{u \la v}), \\
            & Y_{u \ra v} = O_{u \ra v} \cdot R_{u \ra v}, \\
            & Y_{u \la v} = Y_{u \ra v} \cdot O_{u \la v} \cdot R_{u \la v}.
        \end{align*}
\end{enumerate}

\begin{table*}
\caption{Methods comparison}
\label{table:methods-comparison}
\centering
    \begin{tabular}{ccc} \hline
       Methods & Training Loss & Validation Metric \\ \hline
       Conventional & $l$ & $\sum_{v \in V_u} \lambda(\rank_{u}(v)) \cdot \left(2^{Y_{u \ra v} + Y_{u \la v}} - 1\right)$ \\
       IPW1 & $l_{\mathrm{IPW1}}$ &  $\sum_{v \in V_u} \lambda(\rank_{u}(v)) \cdot \frac{\left(2^{Y_{u \ra v} + Y_{u \la v}} - 1\right)}{\theta_{u \ra v}}$ \\
       IPW2 (proposed) & $l_{\mathrm{IPW2}}$ & $\sum_{v \in V_u} \lambda(\rank_{u}(v)) \cdot h_{\mathrm{IPW}}(Y_{u \ra v}, Y_{u \la v})$ \\
       \bottomrule
    \end{tabular}
\end{table*}

\begin{table*}[ht]
\caption{Test score of all models for each fold. For each fold and setting, the best and worst results are denoted in bold and underlined, respectively.}
\label{table:result1}
\centering
    \begin{tabular}{rrrr|rrr|rrr|rrrr} \hline
       DCG & \multicolumn{3}{c}{@3} & \multicolumn{3}{c}{@10} & \multicolumn{3}{c}{@20} & \multicolumn{3}{c}{@30}  \\ 
       fold & baseline & IPW1 & IPW2 & baseline & IPW1 & IPW2 & baseline & IPW1 & IPW2 & baseline & IPW1 & IPW2 \\ \hline
       1 & 0.9270 & \underline{0.7900} & \bf{0.9790} & 2.3565 & \underline{1.8320} & \bf{2.3729} & 3.5849 & \underline{2.7561} & \bf{3.6886} & 5.1954 & \underline{4.8274} & \bf{5.2256} \\
       2 & 0.6972 & \underline{0.6705} & \bf{0.8885} & \underline{1.8785} & 1.9064 & \bf{2.0376} & \underline{3.0666} & 3.1751 & \bf{3.3240} & \underline{4.9080} & 4.9642 & \bf{5.0639} \\ 
       3 & \underline{0.6144} & 0.6587 & \bf{0.8769} & 1.6822 & \underline{1.6650} & \bf{1.9998} & \underline{2.7682} & 2.7702 & \bf{2.9078} & \underline{4.3774} & 4.3975 & \bf{4.5385} \\ 
       4 & \bf{1.2125} & 1.0749 & \underline{1.0014} & 2.7112 & \bf{2.7789} & \underline{2.2095} & \bf{4.1889} & 4.1606 & \underline{3.3233} & \bf{6.4263} & 6.4055 & \underline{6.0450} \\ 
       5 & \bf{0.9372} & \underline{0.8957} & 0.9127 & 2.3912 & \bf{2.4353} & \underline{2.3590} & 3.6510 & \bf{3.6809} & \underline{3.6798} & 5.5462 & \bf{5.5516} & \underline{5.5205} \\ 
       \hline
    \end{tabular}
\end{table*}

\begin{table*}[ht]
\caption{Test score for each hyperparameter $\eta$. The results are highlighted in the same way as \Cref{table:result1}.}
\label{table:result2}
\centering
    \begin{tabular}{rrrr|rrr|rrr|rrrr} \hline
       DCG & \multicolumn{3}{c}{@3} & \multicolumn{3}{c}{@10} & \multicolumn{3}{c}{@20} & \multicolumn{3}{c}{@30} \\
       $\eta$ & baseline & IPW1 & IPW2 & baseline & IPW1 & IPW2 & baseline & IPW1 & IPW2 & baseline & IPW1 & IPW2 \\ \hline
       0.5 & 0.9270 & \underline{0.7900} & \bf{0.9790} & 2.3565 & \underline{1.8320} & \bf{2.3729} & 3.5849 & \underline{2.7561} & \bf{3.6886} & 5.1954 & \underline{4.8274} & \bf{5.2256} \\
       0.6 & 1.0835 & \bf{1.1663} & \underline{0.9171} & 2.4787 & \bf{2.5708} & \underline{2.0815} & 3.7181 & \bf{3.8070} & \underline{3.2703} & 5.3004 & \bf{5.3813} & \underline{5.0614} \\
       0.8 & 1.0004 & \underline{0.8008} & \bf{1.0619} & 2.2892 & \underline{1.7519} & \bf{2.4417} & 3.4547 & \underline{2.6889} & \bf{3.7024} & 5.1839 & \underline{4.8023} & \bf{5.2809} \\ 
       1.0 & 0.9048 & \underline{0.8916} & \bf{1.0209} & 2.2697 & \underline{1.9884} & \bf{2.4283} & 3.5770 & \underline{2.9147} & \bf{3.6907} & 5.1688 & \underline{4.9179} & \bf{5.2619} \\
       \bottomrule
    \end{tabular}
\end{table*}
\subsubsection{Methods}
We used Matrix Factorization-based models to predict mutual preferences. To represent each slide of preferences, we prepared two kinds of embedding space $W_{pro}$ and $W_{rea}$. 
For a pair $(u, v), u \in U, v \in V$, we define a mutual preference score $s_{u v}$ as follows:
\begin{enumerate}
    \item $s_{u \ra v} = \sigma(w_u^\top w_v)$, where $w_u, w_v \in W_{pro}$,
    \item $s_{u \la v} = \sigma(\tilde{w}_u^\top \tilde{w}_v) $, where $\tilde{w}_u, \tilde{w}_v \in W_{rea}$,
    \item $s_{uv} = s_{u \ra v} \cdot s_{u \la v}$,
\end{enumerate}
where $\sigma(\cdot)$ is the sigmoid function.
As an intuitive explanation, $s_{u \ra v}$ represents the preference score from $u$ to $v$, and $s_{u \la v}$ represents the preference score from $v$ to $u$.
In conducting the experiment, we fixed the dimension of the embedding space at 64.

We formulate the training loss related to a certain user $u$'s instances as
$$
l = -\sum_{v \in V_{u}} y_{u \ra v} \log \frac{s_{u \ra v}}{\sum_{v' \in V_{u}} s_{u \ra v'}} -\sum_{v \in V_{u}} y_{u \la v} \log \frac{s_{u \la v}}{\sum_{v' \in V_{u}} s_{u \la v'}},
$$
where the first term corresponds to the learning of the proactive side's preferences, and the second term corresponds to the reactive side's.
Also, two IPW versions of the list-wise loss (one for the conventional IPW loss and the other for two-sided IPW) can be expressed as follows:

\begin{align*}
    l_{\mathrm{IPW1}} &= -\sum_{v} \frac{y_{u \ra v}}{\theta_{u \ra v}} \log \frac{s_{u \ra v}}{\sum_{v'} s_{u \ra v'}} -\sum_{v} \frac{y_{u \la v}}{\theta_{u \ra v}} \log \frac{s_{u \la v}}{\sum_{v'} s_{u \la v'}}, \\
    l_{\mathrm{IPW2}} &= -\sum_{v} \frac{y_{u \ra v}}{\theta_{u \ra v}} \log \frac{s_{u \ra v}}{\sum_{v'} s_{u \ra v'}} -\sum_{v} \frac{y_{u \la v}}{\theta_{u \ra v}\theta_{u \la v}} \log \frac{s_{u \la v}}{\sum_{v'} s_{u \la v'}}. \\
\end{align*}

During the training phase, we saved the best model with the best corresponding ranking metric best for evaluating the test data. 
\Cref{table:methods-comparison} summarises the training loss and validation metric used in each method.

For simplicity, we used the propensity score $\theta_{u \ra v}, \theta_{u \la v}$ that was employed for the data preparation phase in the training and validation phases as well.

\subsubsection{Evaluation}
Since we are able to access the ground-truth relevance of each pair, we evaluated the models by the discounted cumulative gain (DCG) using the true relevance and the ranking deduced by the score $s_{uv}$. 
For a given positive integer $K$, the DCG score for the user $u$ can be calculated as 

$$
DCG@K(u) = \sum_{i=1}^{K} \frac{2^{R_{u \ra v_i} (1 + R_{u \la v_i})}}{\log_2\left(i + 1\right)},
$$
where $v_i$ is the $i$-th element in the ranking induced by the score.

\subsection{Results \& Discussions}
We first examined each method's test data on each fold.
Here we fix the hyperparameter $\eta=0.5$ and evaluate the $DCG@K$ for varying $K \in \left\{3, 10, 20, 30\right\}$.

The results are shown in Table \ref{table:result1}. From the table, it is verified that in more than half of the cases, the proposed method gives the best performance among the three methods. 

Next, we change the hyperparameter $\eta \in \left\{0.5, 0.6, 0.8, 1.0 \right\}$ and evaluate the test score on the fold-1 data.
The results are shown in Table \ref{table:result2}.
Table \ref{table:result2} shows that the proposed method was relatively robust to the data generation process.

It is worth noticing that IPW1 gave the worst test scores in many cases.
It appears that in IPW1, the emphasis placed on learning from the proactive side's feedback may have inadvertently resulted in a deterioration in the learning from the reactive side's feedback.

% 時間があったら多様性が生まれたか、rareなアイテムに限定した評価指標が改善してるか実験を回す

\section{Conclusion}
Learning-to-rank (LTR) algorithms have become essential in modern recommendation systems. 
However, achieving \emph{unbiased} LTR is a significant challenge, particularly when using implicit feedback like clicks, which can be influenced by factors other than user preferences. 
Existing unbiased LTR techniques mainly focus on single-sided markets, which overlook two-sided markets where mutual preference matching determines final conversion. 
In such two-sided market platforms, position biases can emerge from both user sides.
To address such position biases, we proposed a specialized unbiased LTR framework for two-sided matching platforms. 
We introduced the two-sided inverse probability weighting estimator and showed its unbiasedness in two-sided matching platforms. 
Our experiments demonstrated that our method outperformed  existing unbiased LTR methods in terms of prediction accuracy and robustness.

\bibliographystyle{ACM-Reference-Format}
\bibliography{sample-base}

%\appendix

% \section{}

\end{document}